\documentclass[11pt]{article}
\usepackage{enumerate}
\usepackage[backref, colorlinks,citecolor=blue,bookmarks=true]{hyperref}
\usepackage{graphicx}
\usepackage{mathrsfs,amssymb,amsmath,textcomp,amsfonts,amsthm,bbm,fullpage}
\usepackage{color}
\usepackage{multirow,amsmath, dsfont, amscd, latexsym,color,amssymb,setspace}
\usepackage{ifthen}
\usepackage{enumerate}
\usepackage{mathtools}
\usepackage{fullpage}
\usepackage{graphicx}
\usepackage{float}

\usepackage[usenames,dvipsnames]{xcolor}

\newtheorem{theorem}{Theorem}
\newtheorem{conjecture}[theorem]{Conjecture}
\newtheorem{definition}[theorem]{Definition}
\newtheorem{lemma}[theorem]{Lemma}
\newtheorem{proposition}[theorem]{Proposition}

\newtheorem{example}[theorem]{Example}

\newtheorem{Thm}[theorem]{Theorem}
\newtheorem{Lem}[theorem]{Lemma}
\newtheorem{Cor}[theorem]{Corollary}

\newtheorem{Def}[theorem]{Definition}

\newtheorem*{Conj*}{Conjecture}

\newenvironment{proofof}[1]{\begin{trivlist} \item {\bf Proof
#1:~~}}
  {\qed\end{trivlist}}

\newcommand{\mc}[1]{\ensuremath \mathcal{#1}}
\newcommand{\tmn}{\ensuremath T_{m \times n}}
\newcommand{\eat}[1]{}

\def\colorful{1}

\ifnum\colorful=1

\fi

\ifnum\colorful=0

\fi

\newcommand{\ignore}[1]{}

\newcommand{\C}{\mathcal{C}}

\newcommand{\F}{\mathbb{F}}

\newenvironment{Proof}{\medbreak
\noindent {\bf Proof:~}}{\unskip\nobreak\hfill\hskip 2em \qed\par\medbreak}

\begin{document}

\title{Maximally Recoverable Codes for Grid-like Topologies}

\author{Parikshit Gopalan \\ VMware Research \\ pgopalan@vmware.com
\and Guangda Hu \\ Princeton University \\ guangdah@cs.princeton.edu
\and Swastik Kopparty \\ Rutgers University \\ swastik.kopparty@gmail.com
\and Shubhangi Saraf \\ Rutgers University \\ shubhangi.saraf@gmail.com
\and Carol Wang \\ National Univ. of Singapore \\ elecaro@nus.edu.sg
\and Sergey Yekhanin \\ Microsoft Research \\ yekhanin@microsoft.com}
\date{}
\maketitle

\thispagestyle{empty}

\begin{abstract}

The explosion in the volumes of data being stored online has resulted in distributed storage systems transitioning to erasure coding based schemes. Yet, the codes being deployed in practice are fairly short. In this work, we address what we view as the main coding theoretic barrier to deploying longer codes in storage: at large lengths, failures are not independent and correlated failures are inevitable. This motivates designing codes that allow quick data recovery even after large correlated failures, and which have efficient encoding and decoding.

We propose that code design for distributed storage be viewed as a two step process. The first step is choose a {\it topology} of the code, which incorporates knowledge about the correlated failures that need to be handled, and ensures local recovery from such failures. In the second step one specifies a code with the chosen topology by choosing coefficients from a finite field $\mathbb{F}_q$. In this step, one tries to balance reliability (which is better over larger fields) with encoding and decoding efficiency (which is better over smaller fields).

This work initiates an in-depth study of this reliability/efficiency tradeoff. We consider the field-size needed for achieving {\it maximal recoverability}: the strongest reliability possible with a given topology. We propose a family of topologies called grid-like topologies which unify a number of topologies considered both in theory and practice, and prove the following results about codes for such topologies:

\begin{itemize}
\item The first super-polynomial lower bound on the field size needed for achieving maximal recoverability in a simple grid-like topology. To our knowledge, there was no super-linear lower bound known before, for any topology.

\item A combinatorial characterization of erasure patterns correctable by Maximally Recoverable codes for a topology which corresponds to tensoring MDS codes with a parity check code. This topology is used in practice (for instance see \cite{MLRH14}). We conjecture a similar characterization for Maximally Recoverable codes instantiating arbitrary tensor product topologies.
\end{itemize}

\end{abstract}

\newpage
\setcounter{page}{1}

\section{Introduction} \label{sec:intro}

The explosion in the volumes of data being stored online means that
duplicating or triplicating data is not economically feasible. This
has resulted in distributed storage systems employing erasure coding
based schemes in order to ensure reliability with low storage
overheads. Spurred by this, there has been an active line of research
in coding theory focusing on distributed storage. Two main paradigms
have emerged from this work: local reconstruction~\cite{GHSY,XOR_ELE}
and local regeneration~\cite{Dimakis_1}, both focusing on the
efficiency of the decoder in typical failure scenarios (which in
storage means one or a few machines being unavailable, perhaps
temporarily). The former focuses on the number of disk reads needed to
handle such failures, the latter on the amount of communication. In
the last few years, the theory around these codes has developed
rapidly. There are constructions known that achieve optimality for
various tradeoffs. A number of these codes have been deployed at scale
in the real world~\cite{HuangSX,MLRH14}.

Yet, the length of codes being used for data storage thus far is quite small: often in the low double digits.
The coding-theoretic motivation for moving to larger lengths is
obvious: coding at larger lengths allows better error-tolerance for a
given overhead. There is also ample practical motivation, coming from
the need to reduce storage costs. Increasingly, data stored in the
cloud are geographically distributed across data centers, so that even
if one location is offline for some time, data are still
accessible. The simple solution of replication across data centers is
expensive, and can nullify the gains from erasure coding within a data
center. Finally, historical trends suggest that the transition to
longer length codes should happen eventually. Thus it is important to
understand what the current barriers to using longer codes are.

In this work, we address what we view as the main coding theoretic
barrier to deploying longer codes in storage: at large lengths, the
assumption that various nodes fail independently is just not true,
correlated failures are inevitable. This motivates the task of {\em
  designing codes that allow quick data recovery even after large
  correlated failures, and which have efficient encoding and
  decoding.}

\subsection{Codes with a topology}

The coding theoretic challenges arising in distributed storage setting are very different from those encountered when codes are used for transmission, or even storage on a single  device. There are two main reasons behind it:

\begin{itemize}
\item {\bf Correlated failures:} In distributed storage, at large lengths, one cannot assume that individual codeword coordinates fail independently. One has to explicitly deal with large correlated failures, which might have different sources e.g. a rack failure, a datacenter failure, a simultaneous upgrade applied to a large group of machines, or failure of a power source shared by multiple machines. The structure of such correlated failures varies with deployment but is typically known at the code design stage, and  can be incorporated in the code layout.

\item {\bf The need for locality:} Locality addresses the challenge of efficiently serving requests for unavailable data and maintaining an erasure encoded representation. In particular, when a node or a correlated group of nodes fails, one has to be able to quickly reconstruct every lost node in order to keep the data readily available for the users and to maintain the same level of redundancy in the system. We say that a certain packet has {\it locality} $r$ if it can be recovered from accessing only $r$ other packets (think of $r$ as being much less than the codeword length). We would like to ensure locality for typical failure scenarios. At short lengths with independent failures, a single or a few failures might be a reasonable model for what is typical. But at longer lengths, we would like locality after correlated failures (like a data center being offline), which might mean that a constant fraction of machines is unavailable.
\end{itemize}

As a result, in designing codes for distributed storage one tries to incorporate knowledge about correlated failures in the design, in a way that guarantees efficient handling of such failures.
The kinds of code construction problems that arise from this are different from those in classical coding theory, but we feel they are ripe for theoretical analysis. To facilitate this, we propose viewing the design of erasure codes for distributed storage as a two step process, where we intentionally separate out incorporating real-world knowledge about correlated failure patterns from code specification, which is very much within of realm of coding theoretic techniques.

\begin{enumerate}

\item {\bf Picking a topology: } The first step is to determine the {\it topology} of the code, driven by the particular collection of correlated failures that need to be handled. Informally, one can think of a topology as specifying the {\em supports} for the parity check equations, but not the coefficients (or even which field they lie in). The topology specifies the number of redundant symbols and the data symbols that each of them depends on. This can be used to ensure the existence of short linear dependencies between specific codeword coordinates, so that the desired locality property holds for the correlated failure patterns of interest. This is the step which incroporates real-world knowledge about likely correlated failures into the design.

\item {\bf Specifying coefficients: } In the second step one explicitly specifies a code with the chosen topology.   We choose coefficients from a finite field $\mathbb{F}_q$, which fixes the  redundant symbols as explicit $\F_q$-linear combinations of data symbols respecting the dependency constraints from the previous stage. This step typically utilizes tools from classical coding theory, but the objectives are different:

\begin{itemize}
\item {\bf Optimizing encoding/decoding efficiency: } Encoding  a linear code and decoding it from erasures involve matrix vector multiplication and linear equation solving respectively. Both of these require performing numerous finite field arithmetic operations. Having small finite fields results in faster encoding and decoding and thus improves the overall throughput of the system \cite[Section 2]{Plank}. In theory, field sizes which scale polynomially in the codeword length are desirable. Coefficient sizes of a few bytes are preferred in practice.

\item {\bf Maximizing reliability: } Worst-case distance or the number of random failures tolerated are unsatisfactory reliability measures for codes with a prescribed topology. The notion of maximal recoverability first proposed by~\cite{CHL} and generalized by~\cite{GHJY} provides a conceptually simple answer to the question {\em what is the best code for a given topology?}. Once we fix a topology and a set of erasures, decoding reduces to solving a system of linear equations. Maximal recoverability requires that the code corrects every failure pattern which is permissible by linear algebra, given the topology. Equivalently, a Maximally Recoverable (MR) code corrects every erasure pattern that is correctable for some fixing of coefficients in the same topology.
\end{itemize}

\end{enumerate}

The current evidence suggests that it is generally hard to achieve both small field size and maximal recoverability
simultaneously. Reed Solomon codes are the one notable exception to this rule, they are maximally recoverable codes for the trivial topology, and they have a linear field size (which is known to be
optimal up to constant factors). Analogous results are not known even
in topologies that are only slightly more complex than Reed-Solomon (see for instance \cite{GHJY}). For arbitrary topologies,
random codes are maximally recoverable but over fields of exponential
size, and often nothing better is known.

This points at a possible tradeoff between these two
requirements. This tradeoff is the main subject of our work. It may be
the case that in some topologies, the field-sizes required to achieve
maximal recoverability are prohibitively large, so one needs to pick a
different point on the tradeoff curve. A starting point for exploring
this tradeoff is to understand the failure patterns that can be
corrected by maximally recoverable codes for a topology, a problem
that can again be challenging even in simple settings. Given this
discussion, we propose the following questions as the natural main
goals in the study of maximal recoverability.

For a given topology
\begin{itemize}
\item Determine the smallest field size over which MR codes exist.
\item Characterize the failure patterns that can be corrected by MR codes.
\item Find explicit constructions of MR codes over small fields.
\end{itemize}

In theory one could ask these questions about any topology, but the
important topologies are simple ones which model how machines tend to
be laid out within and across data centers. In this work, we propose a
family of topologies called grid-like topologies which unify a number
of topologies considered both in coding theory and practice. In short,
codes with grid-like topologies can be viewed as tensor products of
row and column codes, augmented with  global parity check
constraints. They provide a unified framework for MDS codes, tensor
product codes, LRCs and more (see the discussion in Section
\ref{Sec:Prelim}).

We prove the following results about codes for grid-like topologies:

\begin{itemize}
\item The first super-polynomial lower bound on the field size needed for achieving maximal recoverability in any topology (in fact our bound applies to a very simple grid-like topology).

\item A combinatorial characterization of erasure patterns correctable by Maximally Recoverable codes for a topology which corresponds to tensoring MDS codes with a parity check code. This topology is used in practice (for instance see Facebook's f4 storage system \cite{MLRH14}).

\item A new asymptotically optimal family of Maximally Recoverable codes for a certain basic topology giving an alternative proof to a result of~\cite[Theorem 2.2]{Blaum}.
\end{itemize}

\subsection{Outline}
Section~\ref{Sec:Prelim} gives a formal definition of
grid-like topologies and explains why that definition captures
the needs that arise in distributed storage. In Section~\ref{Sec:Results} we present formal statements of our three
main theorems: the lower bound for alphabet size of MR codes, the
combinatorial classification of erasure patterns correctable by MR
codes, and an upper bound for the alphabet size of MR codes.

In Section ~\ref{Sec:lb} we establish our alphabet size lower bound. In
Section~\ref{Sec:Class} we obtain the classification result. In
Section~\ref{Sec:H2} we give our new construction of maximally
recoverable codes.  We survey more related work in
Appendix \ref{Sec:Work}. In section~\ref{Sec:Open} we discuss the key
questions that remain open.

\subsection{Notation}\label{Sec:Notation}
We use the following standard mathematical notation:
\begin{itemize}
\item $[s]=\{1,\ldots,s\};$

\item Let ${\bf w}\in \mathbb{F}^n$ be a vector. Let $\mathrm{supp}({\bf w})\subseteq [n]$ denote the set of non-zero coordinates of ${\bf w}.$

\item $[n,k,d]_q$ denotes a linear code (subspace) of dimension $k,$ codeword length $n,$ and distance $d$ over a field $\mathbb{F}_q.$ We often write $[n,k,d]$ instead of $[n,k,d]_q$ when the particular choice of the field is not important.

\item Let $C$ be an $[n,k,d]$ code and $S\subseteq [n],$ $|S|=k.$ We say that $S$ is an information set if the restriction $C|_S=\mathbb{F}_q^k.$

\item An $[n,k,d]$ code is called Maximum Distance Separable (MDS) if $d=n-k+1.$ MDS codes have many nice properties. In particular an $[n,k,d]$ code is MDS if and only if every subset of its $k$ coordinates is an information set. Alternatively, an $[n,k,d]$ code is MDS if and only if it corrects any collection of $(n-k)$ simultaneous erasures~\cite{MS}.

\item Let $C_1$ be an $[n_1,k_1,d_1]$ code and $C_2$ be an $[n_2,k_2,d_2]$ code. The tensor product $C_1\otimes C_2$ is an $[n_1n_2,k_1k_2,d_1d_2]$ code where the codewords of $C_1\otimes C_2$ are matrices of size $n_1\times n_2,$ where each column belongs to $C_1$ and each row belongs to $C_2.$ If $U\subseteq [n_1]$ is an information set of $C_1$ and $V\subseteq [n_2]$ is an information set of $C_2;$ then $U\times V$ is an information set of $C_1\otimes C_2,$ e.g.,~\cite{MS}.
\end{itemize}

\section{Grid-like topologies}
\label{Sec:Prelim}

We will restrict our attention to fields of characteristic $2$,
the natural setting for storage. We propose studying maximal recoverability for a simple class of
topologies called grid-like topologies that unify and generalize many
of the layouts that are used in practice~\cite{HuangSX, MLRH14}.
We specify topologies via dual
constraints. This way of defining topologies simplifies the proofs,
however it might not be immediately clear that topologies defined like
that indeed capture the needs that arise in distributed storage. We
explain the connection in~Proposition~\ref{Prop:DualToPrimal}.

\begin{definition}\label{Def:Topology} (Grid-like topology)
Let $m\leq n$ be integers. Consider an $m \times n$ array of symbols
$\{x_{ij}\}_{i \in [m], j \in [n]}$ over a field $\F$ of
characteristic $2.$ Let $0\leq a\leq m-1,$ $0\leq b\leq n-1,$ and
$0\leq h\leq (m-a)(n-b)-1.$ Let $T_{m \times n}(a,b,h)$ denote the
topology where there are $a$ parity check equations per column, $b$
parity check equations per  row, and $h$ global parity check equations
that depend on all symbols. A code with this topology is specified by
field elements $\left\{\alpha_i^{(k)}\right\}_{i \in [m],k \in [a]}$,
$\left\{\beta_j^{(k)}\right\}_{j \in [n],k \in [b]}$ and
$\left\{\gamma_{ij}^{(k)}\right\}_{i \in [m], j \in [n],k \in [h]}$.
\begin{enumerate}
\item Each column $j \in [n]$ satisfies  the constraints
\begin{equation}\label{Eqn:Top1}
\sum_{i=1}^{m}\alpha_i^{(k)}x_{ij} = 0 \quad  \forall k \in [a].
\end{equation}
\item Each row $i \in [m]$ satisfies the constraints:
\begin{equation}\label{Eqn:Top2}
\sum_{j=1}^{n}\beta_j^{(k)}x_{ij} = 0 \quad  \forall k \in [b].
\end{equation}
\item The symbols satisfy $h$ global constraints given by
\begin{equation}\label{Eqn:Top3}
\sum_{i=1}^m\sum_{j=1}^{n}\gamma_{ij}^{(k)}x_{ij} = 0 \quad  \forall k \in [h].
\end{equation}
\end{enumerate}
A setting of $\{\alpha_i^{(k)}\}, \{\beta_j^{(k)}\},\{\gamma_{ij}^{(k)}\}$ from a field $\F \supseteq \F_2$ specifies a code $\C = \C(\{\alpha_i^{(k)}\},\{\beta_j^{(k)}\},\{\gamma_{ij}^{(k)}\})$ that instantiates the topology $T_{m \times n}(a,b,h)$.
\end{definition}

Intuitively, constraints~(\ref{Eqn:Top1}) above ensure that that there are local dependencies in every column; constraints~(\ref{Eqn:Top2}) ensure that that there are local dependencies in every row; and constraints~(\ref{Eqn:Top3}) provide additional reliability guarantees, if the guarantees provided by~(\ref{Eqn:Top1}) and~(\ref{Eqn:Top2}) alone are deemed not sufficient. In most settings of interest $a,b,$ and $h$ are rather small compared to $m$ and~$n.$

In what follows we refer to constraints~(\ref{Eqn:Top1}) as specifying a code $C_{\mathrm{col}}\subseteq \mathbb{F}^m,$ and to constraints~(\ref{Eqn:Top2}) as specifying a code $C_{\mathrm{row}}\subseteq \mathbb{F}^n.$ When $h =0$, the resulting code is exactly $C_{\mathrm{col}} \otimes C_{\mathrm{row}}$. For larger $h$, we can view $\C$ as a subspace of $C_{\mathrm{col}} \otimes C_{\mathrm{row}}$ with co-dimension $h$.

\begin{definition}\label{Def:CorrectablePatt}
 A failure pattern is a set $E \subseteq [m] \times [n]$ of symbols that are erased. Pattern $E$ is correctable for the topology $T_{m \times n}(a,b,h)$ if there exists a code instantiating the topology where the variables $\{x_{ij}\}_{(i,j) \in E}$ can be recovered from the parity check equations.
\end{definition}

\begin{definition}\label{Def:MR}
A code $\C$ that instantiates the topology $T_{m \times n}(a,b,h)$ is Maximally Recoverable (MR) if it corrects every failure pattern that is correctable for the topology.
\end{definition}
In other words a code that instantiates a topology is maximally recoverable, if it corrects all erasure patterns that are information theoretically correctable given the topology (dependency) constraints. We also note that \cite{GHJY} define the notion of a topology and maximal recoverabilty in full generality. Since our focus here is only on grid-like topologies, we refer the curious reader to that paper for the general definition.
We now state a basic proposition about such codes, the proof is in Appendix \ref{sec:app}.

\begin{proposition}\label{Prop:DualToPrimal}
Let $C$ be an MR instantiation of the topology $T_{m\times n}(a,b,h).$ We have
\begin{enumerate}
\item The dimension of $C$ is given by
      \begin{equation}\label{Eqn:MrDim}
      \dim C = (m-a)(n-b)-h.
      \end{equation}
      Moreover,
      \begin{equation}\label{Eqn:MrDimRowCol}
      \dim C_{\mathrm{col}} = m-a \quad\mathrm{and}\quad  \dim C_{\mathrm{row}} = n-b.
      \end{equation}

\item Let $U\subseteq [m],$ $|U|=m-a$ and $V\subseteq [n],$ $|V|=n-b$
  be arbitrary. Then $C|_{U\times V}$ is
  an $$[(m-a)(n-b),(m-a)(n-b)-h,h+1]$$ MDS code. Any subset
  $S\subseteq U\times V,$ $|S|=(m-a)(n-b)-h$ is an information set.

\item Assume
  \begin{equation}\label{Eqn:TechLoc}
    h\leq (m-a)(n-b)-\max\{(m-a),(n-b)\},
      \end{equation}
      then the code $C_{\mathrm{col}}$ is an $[m,m-a,a+1]$ MDS code
      and the code $C_{\mathrm{row}}$ is an $[n,n-b,b+1]$ MDS
      code. Moreover, for all $j\in [n],$ $C$ restricted to column $j$
      is the code $C_{\mathrm{col}}$ and for all $i\in [m],$ $C$
      restricted to row $i$ is the code $C_{\mathrm{row}}.$
\end{enumerate}
\end{proposition}

Let us use this proposition to see why grid-like topologies arise naturally in distributed storage.
Consider a setting where we have $m$ datacenters, each with $n$ machines where $m \ll n$. One can use
a code instantiating $T_{m\times n}(a,b,h)$ to distribute data across the datacenters.

\begin{itemize}
\item The code is systematic, item (2) tells us how to select the information symbols. So when no failures happen, the data are readily accessible.

\item When up to $a$ datacenters are unavailable, each packet can be recovered by accessing at most $(m-a)$ symbols across the remaining datacenters using the column MDS code. This involves cross datacenter traffic, but since $m$ is small,  we do not need to read too many symbols. When fewer than $a$ data centers are offline, the MDS property implies that any $a$ packets being received suffices for successful decoding.

\item When up to $b$ failures happen within a particular datacenter, every packet can be accessed by performing $(n-b)$ reads within the datacenter using the row MDS code.

\item The $h$ global parities improve the worst case distance of the code. They are only used when row and column decoding fails, since using them in decoding involves using all the code symbols and hence requires a lot of communication across datacenters.
\end{itemize}

This shows that grid-like topologies give a way to guarantee local recovery
after certain kinds of correlated failures. Depending on the precise kind of
correlated failures and the locality guarantees desired, there might
be other ways. But in certain (restricted but important) settings,
one can uniquely specify a topology as the only possible solution. Suppose our goal is to
provide good locality for all data symbols after one machine failure, and failure tolerance
for up to $h$ failures.\cite[Theorem 9]{GHSY} showed that for some
parameter settings, all optimal length codes that guarantee this must
have a fixed topology (which is not grid-like but is closely related to $T(1,0,h)$).

More generally, grid-like topologies do provide a unifying
framework for several settings that arise in practice and have
been studied in the literature (see examples below).

\begin{enumerate}

\item The topology $T_{m \times n}(1,0,h)$
  has received a considerable
  amount of attention, especially in the recent work on LRCs~\cite{BHH,GHSY,TB,GHJY,BK,LL,HY,BPSY}.
  Correctable patterns for
  this topology are fully characterized~\cite{BHH,GHJY}. The best
  known constructions~\cite{GHJY} are slightly better than random
  codes in their alphabet size, but are still far from polynomial in
  the block-length. A construction of codes over fields of linear size
  with the largest possible distance for $T_{m\times n}(1,0,h)$ (a
  weaker property than maximal recoverability) has been given
  in~\cite{TB}.

\item A maximally recoverable code instantiating a topology closely related to $T_{2\times 7}(0,1,2)$ is used by Microsoft's Azure storage~\cite{HuangSX}.

\item A code instantiating $T_{3\times 14}(1,4,0)$ is used by Facebook
  in its f4 storage system~\cite{MLRH14}. The code is the tensor
  product of a Reed-Solomon code within data centers with a parity
  check code across data centers.

\item Maximum distance separable (MDS) codes can be viewed as
  maximally recoverable codes for $T_{1 \times n}(0,0,h)$. Reed
  Solomon codes yield explicit constructions over an alphabet of size
  $n,$ and there are lower bounds of $\Omega(n)$ on the field
  size~\cite{MainMDS1,MainMDS2}.

\end{enumerate}

\section{Formal statements of our results}\label{Sec:Results}

\subsection{A super-polynomial field-size lower bound}

Our main result is the first super-polynomial lower bound on the field
size required for maximally recoverable codes.
Previously, there was not even a {\em super-linear} lower bound known,
for {\em any} topology. Our lower bounds apply to all topologies $T_{m \times n}(a,b,h)$
where $a \geq 1, b \geq 1, h \geq 1$, and are meaningful as long as $h \ll \min(m-a,n-b)$. If
we think of the setting where $m = n$ are growing and $a,b,h =O(1)$,
then our lower bound is $\exp(\Omega((\log(n))^2))$.

\begin{Thm}
\label{Thm:LBgen}
Assume that $a \geq 1, b\geq 1, h \geq 1$. Let $n' = \min(m-a
  +1,n-b +1)$ and let $h \leq n'$. Any maximally recoverable code for the topology $T_{m \times
    n}(a,b,h)$ requires field size $q = \exp(\Omega((\log(n'/h))^2))$.
\end{Thm}

The simplest grid-like topology to which this bound applies is $T_{n \times n}(1,1,1)$. Since $T_{n \times n}(1,1,0)$ is
just the parity tensor code, $T_{n \times n}(1,1,1)$ can be viewed as
the parity tensor code with a single global parity check equation
added to it.  Indeed we get our strongest lower
bound for this topology.

\begin{Cor}\label{Cor:LB}
Any maximally recoverable code for the topology $T_{n \times
  n}(1,1,1)$ requires field size $q = \exp(\Omega((\log(n))^2)).$
\end{Cor}

There is also an explicit  construction \cite[Theorem 31]{GHJY} that gives
MR codes for this topology over fields of size $\exp(O(n \log n))$,
matching the probabilitic construction. But this still leaves a gap
between upper and lower bounds.

The key technical ingredient in proving our lower bounds is the following combinatorial lemma, which might be of independent interest.
\begin{lemma}
\label{Lemma:Main}
Let $\gamma:[n] \times [n] \to \mathbb{F}_2^d$ be a labelling of the edges of the complete bipartite graph $K_{n,n}$ by bit vectors such that for any {\em simple} cycle $C$,
\[ \sum_{e \in C}\gamma_e \neq 0.\]
Then we have $d \geq \Omega((\log(n))^2)$.
\end{lemma}

Theorem \ref{Thm:LBgen} clarifies the picture of which grid-like
topologies might admit MR codes with polynomial field sizes. Let us consider the setting when $ a,b,h$ are
$O(1)$ and $m,n$ are growing. Circumventing the lower bound requires at least one of
$a,b,h$ to be $0$. After accounting for symmetries and trivialities,
this leaves two families of grid-like topologies:

\begin{enumerate}
  \item Tensor Products: $T_{m \times n}(a,b,0)$. As we will see in
    the next section, we do not know explicit MR codes for this
    topology (or even a characterization of correctable
    error patterns), but such codes might exist over small fields.

  \item Generalized Locally Recoverable Codes: $T_{m \times n}(a,0,h)$. These
    codes provide locality of $m$ after $a$ erasures, and can tolerate
    $a + h$ worst-case failures (in LRCs, one typically considers $a
    =1$ and $m =O(1)$). This generalization was first defined and
    studied by \cite{PKLV}. For this
    topology, one can extend the constructions in \cite{GHJY}  to
    derive codes with field
    size $(mn)^{O_{a,h}(1)}$. The only lower bound we know is $\Omega(mn)$. Getting a field size which is
    a fixed polynomial in $mn$ is open.
\end{enumerate}

We view resolving these two questions as central to the development of
MR codes, and our next two results make some progress on them.

\subsection{Characterizing correctable erasure patterns in tensor products}

Perhaps the simplest grid-like topology where we cannot characterize correctable
patterns of erasures is the topology $T_{m \times n}(a,b,0),$ which
can be viewed as the tensor of a row and column code, each of which is
MDS (by the last item in Proposition~\ref{Prop:DualToPrimal}). Tensor-product codes are
ubiquitous, especially in storage \cite{RR72}. They are typically
decoded using iterative row-column decoding: if some column has $a$
or fewer erasures, or some row has $b$ for fewer erasures, we decode
it. We repeat until either all erasures have been corrected, or we are
stuck.

When the decoder is stuck, we reach an failure pattern $E$ where
each column in the support of  $E$ has (strictly) more than $a$ erasures, and each row has
more than $b$ erasures. We refer to such patterns as {\em irreducible}.
Recall that by Definiton \ref{Def:CorrectablePatt}, a pattern $E$ is
correctable if one can solve for the missing symbols by applying Gaussian elimination
over all the constraints enforced by the tensor product code ($a$ per
column, $b$ per row). This raises the question: {\em are irreducible patterns uncorrectable?} Or
equivalently, how does iterative low column decoding (which one could
view of as local Gaussian elimination) compare to unrestricted/global
Gaussian elimination? This is a natural question which
has not been addressed previously to our knowledge.

Theorem \ref{Th:Class} implies that there exist irreducible patterns that are
correctable by maximally recoverable codes (see Figure \ref{fig:eg}
for a simple example), hence iterative row-column decoding can be weaker than unrestricted Gaussian elimination.
This raises the question of characterizing which (irreducible) patterns are correctable for MR tensor product
codes. This is the subject of our second result. We present a simple
necessary condition for a pattern to be correctable. We conjecture
that this condition is also sufficient, and prove in in the setting
where $a =1$.

\begin{definition}
\label{Def:RegPattern}
Consider the topology $T_{m\times n}(a,b,0)$ and an erasure pattern $E.$ We say that $E$ is regular if for all $U\subseteq [m],$ $|U|=u$ and $V\subseteq [n],$ $|V|=v$ we have
\begin{equation}\label{Eqn:RegIntersect}
|E\cap \left(U\times V\right)|\leq va+ ub -ab.
\end{equation}
\end{definition}

It is not hard to see that regularity is in fact necessary for
correctability (see Section \ref{sec:app}).

\begin{lemma}\label{Lemma:RegNecess}
If $E$ is not a regular pattern, then it is not correctable for $T_{m\times n}(a,b,0)$.
\end{lemma}

We conjecture that regularity is also sufficient, and thus yields a characterization of the correctable error patterns in $T_{m \times n}(a,b,0)$.

\begin{conjecture}\label{Conj:ClassMain}
An erasure pattern $E$ is correctable for $T_{m\times n}(a,b,0)$ if and only if it is regular.
\end{conjecture}

We prove Conjecture~\ref{Conj:ClassMain} in the restricted case of
$a=1$. This topology of a row code tensored with a parity check code is important in practice \cite{MLRH14}.

\begin{theorem}\label{Th:Class}
A pattern $E$ is correctable for $T_{m\times n}(1,b,0)$ if and only if it is regular for $T_{m\times n}(1,b,0)$.
\end{theorem}

\begin{figure}
\centering
\includegraphics[width=0.3\textwidth]{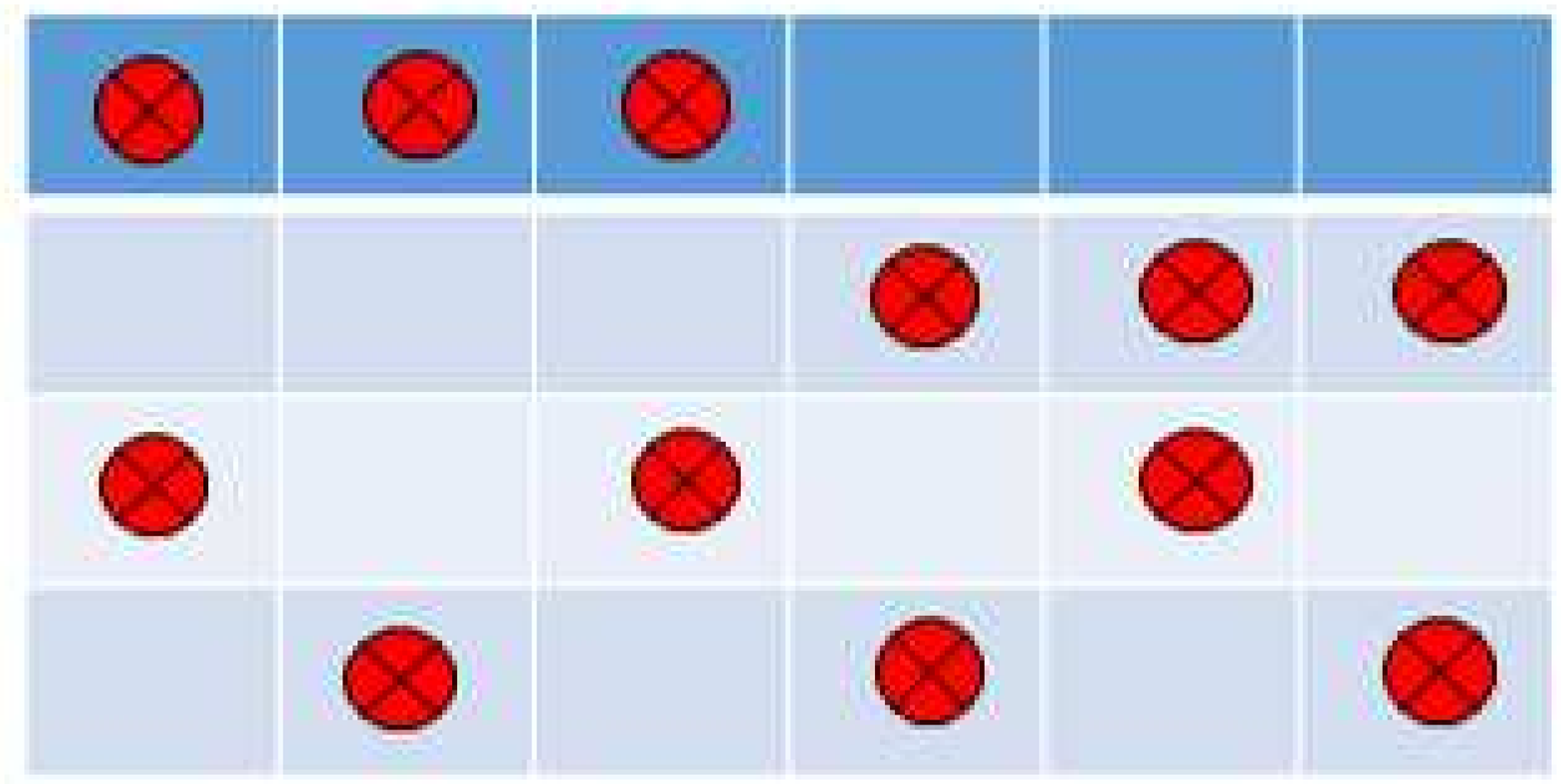}
\caption{Consider the tensor of a $[6,4]$ row code and a $[4,3]$
  column code. This $12$-failure pattern is irreducible, since every row has 3
  erasures and every column has $2$ erasures. By Theorem \ref{Th:Class} it can be
  corrected by an MR tensor product code.}
\label{fig:eg}
\end{figure}

For the failure pattern in Figure \ref{fig:eg} it can be verified that there
exist $[6,4,3]$ MDS codes whose tensor with the $[4,3,2]$ parity code will not correct this error pattern.
The illustrates that the set of correctable failure patterns
for a tensor product code depends on the choice of row and column
codes (in contrast to the row-column decoder). Each of them being MDS
is necessary but not sufficient for maximal recoverability. Finding
explicit maximally recoverable tensor products is an intriguing open
problem, and Conjecture \ref{Conj:ClassMain} might be a good starting point towards it.

\subsection{Asymptotically optimal MR codes for $T_{m\times n}(1,0,2)$}

The topology $T_{m\times n}(1,0,h)$ has received a considerable amount of attention, especially in the recent work on LRCs~\cite{BHH,GHJY,GHSY,BK,LL,HY,BPSY}. For $h=1$ explicit MR codes exist over a field of size $O(m)$ (which is sub-linear in the input length). For $h\geq 2$ there is an $\Omega(mn)$ lower bound on the field size~\cite{GHJY}. In the case of $h=2,$ there is a matching upper bound~\cite[Theorem 2.2]{Blaum}. In what follows we present a new code family re-proving the $O(mn)$ upper bound for $T_{m\times n}(1,0,2).$

\begin{theorem}\label{Th:H2}
For all $m,n$ there exists an explicit maximally recoverable code instantiating the topology $T_{m\times n}(1,0,2)$ over a field of size $O(mn).$
\end{theorem}

\section{Super-polynomial field-size lower bounds}
\label{Sec:lb}

The proof of  Theorem~\ref{Thm:LBgen} combines two main ingredients: the proof
of Lemma~\ref{Lemma:Main}, and a characterization of correctable erasure patterns in $T_{n \times n}(1,1,h)$.
The lower bound for $T_{n \times n}(1,1,1)$ is a fairly direct application Lemma \ref{Lemma:Main}, while lower bounds for $T_{m \times
  n}(a,b,h)$ where $a, b, h \geq 1$ are derived by a reduction to $T_{n \times
  n}(1,1,1)$.

\subsection{The Main Lemma}
\label{Sec:LemmaMain}
In order to prove Lemma \ref{Lemma:Main}, we will consider the more general setting of bounded degree graphs.
We will consider a graph $G=(V,E)$ with maximum degree $D$, where each edge $e\in E$ is assigned a weight $\gamma(e)\in\F$, where $\F$ is a field of characteristic $2.$ Our result will actually apply to any Abelian group. For path $p$ in the graph, we use $\gamma(p)$ to denote the sum of the weights of all edges in the path. Let $P(v_1,v_2)$ ($v_1\neq v_2\in V$) be the set of simple paths from $v_1$ to $v_2$, and $P_k(v_1,v_2)$ be the set of simple paths from $v_1$ to $v_2$ with length $k$, where the {\em length} of a path is the number of edges in that path. For a path in $P_k(v_1,v_2)$, we say that $v_1$ is the 1st vertex, $v_2$ is the $(k+1)$-th vertex, and the other $(k-1)$ vertices on the path are the 2nd through the $k$-th vertices according to their positions. We are interested in graphs with the following property:

\begin{Def} A weighted graph as above satisfies {\bf Property $\mathcal{A}$} if for all $v_1\neq v_2\in V$ and vertex disjoint simple paths $p_1\neq p_2\in P(v_1,v_2)$, their weights satisfy the condition $\gamma(p_1)\neq \gamma(p_2)$.
\end{Def}

It is clear that if we assign weights $\gamma$ to the edges of $K_{m, n}$ such that Equation \eqref{eq:correctable} holds, then Property $\mc{A}$ holds. We next state the main technical lemma of this section, which shows that there cannot be too many paths with the same length and the same weight.
\begin{Lem} \label{lem:main}
If $G$ has Property~$\mathcal{A}$, then for arbitrary vertices $v_1\neq v_2\in V$, positive integer $k\leq\sqrt{D}$, and $\gamma_0\in\Sigma$, the set
$$S=\{p\in P_k(v_1,v_2)\mid \gamma(p)=\gamma_0\}$$
has cardinality at most $k^{\log_2k+1}D^{k-\log_2k-1}$.
\end{Lem}

\noindent {\bf High-level idea:} We think $k$ as a small number for convenience. The goal of the lemma is to show $|S|\lesssim D^{k-\log_2k-1}$. The total number of paths in $|S|$ would be $D^{k-1}$ if all the intermediate $k-1$ vertices could be chosen `freely'. The lemma is saying that we do not have so much `freedom'.

We will show that there is a set $T\subseteq S$ with $|T|\gtrsim|S|/k^2$ such that all paths in $T$ share the same $t$-th vertex for some $t\in[2,k]$. That is saying many paths in $S$ are fixed at the $t$-th vertex, and the choice of the $t$-th vertex is not `free'. Then we fix the prefix before (or suffix after) the $t$-th vertex and recursively apply the argument to the remaining half of the path. Intuitively,  we can do this for $\log_2k$ rounds (since each time we halve the length of the path) and find $\log_2k$ vertices that are not `free', which gives the bound $D^{k-\log_2k-1}$.

The proof is by induction on the length $k$.

\begin{proof}
Let $f(k)=k^{\log_2k+1}D^{k-\log_2k-1}$. We prove by induction on $k$. For $k=1$, we have $|S|\leq1=k^{\log_2k+1}D^{k-\log_2k-1}=f(k)$. Assume we have proved the lemma for lengths up to $k-1$, and we consider the case of $k$ ($2\leq k\leq\sqrt{D}$).

If $S=\emptyset$, the lemma is trivial. We only consider the case that $S\neq\emptyset$. We pick an arbitrary path $p_0\in S$. Then for any other path $p\in S$, $p\neq p_0$, $p$ must intersect $p_0$ at some vertex other than $v_1,v_2$, because of Property~$\mathcal{A}$. That is, there exists $i,j\in[k-1]$ such that the $(i+1)$-th vertex of $p$ is the same as the $(j+1)$-th vertex of $p_0$. Let $T_{ij}$ denote the set of these paths, formally
$$T_{ij} = \{p\in S\setminus\{p_0\}\mid\text{the $(i+1)$-th vertex of $p$ is the $(j+1)$-th vertex of $p_0$}\}.$$
Note that $\bigcup_{i,j\in[k-1]}T_{ij}=S\setminus\{p_0\}$. By the Pigeonhole principle, there must exist $i_0,j_0\in[k-1]$ such that
$$|T_{i_0j_0}|\geq\frac{|S|-1}{(k-1)^2}.$$

We consider the paths in $T_{i_0j_0}$. These paths share the same $(i_0+1)$-th vertex. We denote this vertex by $v_3$. Every path in $T_{i_0j_0}$ can be considered as two parts, the head from $v_1$ to $v_3$ (with length $i_0$) and the tail from $v_3$ to $v_2$ (with length $k-i_0$). We will assume that $i _0 \leq k/2$, so that the head not longer than the tail. If this condition does not hold, we can interchange the definition of head and tail.

\begin{figure}[H]
\centering
\includegraphics{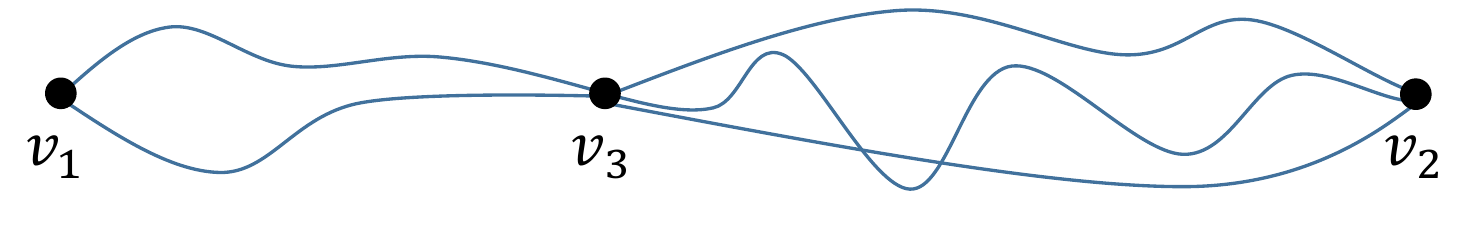}
\caption{Paths in $T_{i_0j_0}$ are fixed at 3 vertices $v_1,v_2,v_3$.}
\end{figure}

The number of possible heads equals the number of simple paths from $v_1$ to $v_3$, which  is at most $D^{i_0-1}$.
We count the paths in $T_{i_0j_0}$ according to their head. For every choice of head, the weight of the tail is fixed because all paths in $T_{i_0j_0}$ have the same total weight. Hence by induction hypothesis, the number of possibilities of the tail for every fixed head is bounded by
\[
f(k-i_0)=(k-i_0)^{\log_2(k-i_0)+1}D^{k-i_0-\log_2(k-i_0)-1}.
\]
So we have
\begin{align*}
|T_{i_0j_0}| &\leq D^{i_0-1}f(k-i_0)\\
& =(k-i_0)^{\log_2(k-i_0)+1}D^{k-\log_2(k-i_0)-2}
\end{align*}
Hence
\begin{align}
|S| & \leq(k-1)^2|T_{i_0j_0}|+1 \nonumber\\
& \leq(k-1)^2(k-i_0)^{\log_2(k-i_0)+1}D^{k-\log_2(k-i_0)-2}+1 \nonumber \\
&\leq k^2(k-i_0)^{\log_2(k-i_0)+1}D^{k-\log_2(k-i_0)-2}. \label{eq1}
\end{align}
Let $t = k - i_0$. Since we assume that $i_0 \leq k/2$, we have $k/2 \leq t \leq k$ and
\[ |S| \leq k^2t^{\log_2(t) +1}D^{k - \log_2(t) -2}.\]
We show that the RHS is at most $f(k)$, by considering the ratio with $f(k)$,
\begin{align*}
\frac{k^2t^{\log_2t+1}D^{k-\log_2t-2}}{k^{\log_2k+1}D^{k-\log_2k-1}}
&= \frac{ t^{\log_2t+1}}{ k^{\log_2k -1}D^{\log_2t-\log_2k +1}}\\
&= \frac{(kt)^{\log_2t-\log_2k+1}}{D^{\log_2t-\log_2k+1} }\ \ \text{since} \ \ t^{\log_2(k)} = k^{\log_2(t)}\\
&= \left(\frac{kt}{D}\right)^{\log_2(2t/k)}\\
& \leq1
\end{align*}
where in the last step we used the fact $t\leq k\leq\sqrt{D}$ and $t\geq k/2$. Thus $|S|\leq f(k)$, hence the claim is proved.
\end{proof}

We proceed to the proof of the main Lemma.

\begin{proofof}{of Lemma~\ref{Lemma:Main}}
We claim that $K_{n,n}$ with weights $\{\gamma(e)\}$ has Property~$\mathcal{A}.$ For $v_1\neq v_2$ and vertex disjoint simple paths $p_1\neq p_2\in P(v_1,v_2)$, we can see that $p_1$ and $p_2$ form a simple cycle. Hence $\gamma(p_1)+\gamma(p_2)\neq 0.$ Since the alphabet of weights has characteristic $2,$ we have $\gamma(p_1)\neq \gamma(p_2),$ and Property~$\mathcal{A}$ is satisfied.

Let $\ell=\lfloor(\sqrt{n}-1)/2\rfloor$ and $k=2\ell+1$. Clearly, $k\leq\sqrt{n}.$ Pick vertices $u$ and $v$ from the two sides of the graph. The number of simple paths from $u$ to $v$ is
$$(n-1)(n-1)(n-2)(n-2)\cdots(n-\ell)(n-\ell)\geq(n-\ell)^{2\ell}.$$
Apply Lemma~\ref{lem:main} for $D=n$. Then for every $\gamma_0\in\F_2^d$, the number of paths from $u$ to $v$ with length $k$ and weight $\gamma_0$ is at most $k^{\log_2k+1}n^{k-\log_2k-1}$. Hence we have
\begin{align*}
2^d \geq & \frac{(n-\ell)^{2\ell}}{k^{\log_2k+1}n^{k-\log_2k-1}}\\
& = \frac{n^{\log_2k}}{k^{\log_2k+1}}\cdot\left(\frac{n-\ell}{n}\right)^{2\ell}\\
& = \frac{n^{\log_2k}}{k^{\log_2k+1}}\cdot \exp(\Theta(\ell^2/n))\\
& =n^{\Omega(\log n)}
\end{align*}
where we used $k=2\ell+1$ and $k,\ell=\Theta(\sqrt{n})$. It follows immediately that $d=\Omega(\log^2n).$
\end{proofof}

\subsection{Characterizing correctable erasure patterns in $T_{n \times n}(1,1,h)$}
\label{sec:t11h}

Recall that the topology $T_{n \times n}(1,1,h)$ is defined by the constraints
\begin{align}
\forall j \in [n], \sum_{i=1}^{m}\alpha_ix_{ij} & = 0, \label{eq:col}\\
\forall i \in [m], \sum_{j=1}^{n}\beta_jx_{ij} & = 0, \label{eq:row}\\
\sum_{i=1}^m\sum_{j=1}^{n}\gamma_{ij}^{(k)}x_{ij} & = 0\quad \forall k\in [h]. \label{eq:global}
\end{align}
An assignment of coefficients specifies a code $\C =
\C(\{\alpha_i\},\{\beta_j\},\{\gamma_{ij}^{(k)}\})$ that instantiates
this topology.  We start by showing that the row and column parity
equations can be taken to be simple XORs without a loss of generality.

\begin{Lem}
\label{lem:1}
Let $m, n \geq 3$ and $h\leq (m-a)(n-b)-\max\{(m-a),(n-b)\}$.
There exists an MR instantiation of $T_{m \times n}(1,1,h)$ where $\alpha_i =1, \beta_j =1$ for every $i \in [m], j \in [n].$
\end{Lem}
\begin{proof}
Consider an arbitrary MR instantiation of $T_{m \times n}(1,1,h).$ By
the choice of $h$, by item (3) in Proposition~\ref{Prop:DualToPrimal} for every $i, j$, $\alpha_i \neq 0$, $\beta_j \neq 0$. Let us define new variables $z_{ij} = \alpha_i\beta_jx_{ij}$. Since $\beta_j^{-1} \in \F^*$ is defined, $\beta_j^{-1}z_{ij} = \alpha_ix_{ij}$.
So we can rewrite~\eqref{eq:col} as
\begin{align}
\notag\forall j \in [n], \ \beta_j^{-1}\cdot \sum_{i=1}^mz_{ij} & = 0 \\
\text{hence} \ \ \forall j \in [n], \   \sum_{i=1}^mz_{ij} & = 0. \label{eq:col1}
\end{align}
Similarly, since $\alpha_i^{-1} \in \F^*$ is defined, we can rewrite~\eqref{eq:row} as
\begin{align}
\sum_{j=1}^nz_{ij} =0.\label{eq:row1}
\end{align}
For each $k\in [h],$ setting $\gamma^{(k)'}_{ij} = \gamma_{ij}^{(k)}\alpha_i^{-1}\beta_j^{-1},$ we can rewrite~\eqref{eq:global} as
\begin{align}
\sum_{ij}\gamma^{(k)'}_{ij}z_{ij} = 0.\label{eq:global1}
\end{align}
It is clear that the code on the $z_{ij}$s defined by~\eqref{eq:col1},~\eqref{eq:row1} and~\eqref{eq:global1} can correct the same set of failures as the original code on the $x_{ij}$s, so the claim is proved.
\end{proof}

For the remainder of this section, we assume that $m,n,h$ satsify the
conditions of Lemma \ref{lem:1}. So we can consider MR
instantiations of $T_{m\times n}(1,1,h)$ where $\alpha_i =1, \beta_j
=1$ for every $i \in [m], j \in [n].$ Such instantiations are
specified by setting the coefficients $\{\gamma_{ij}^{(k)}\},$ $k\in
[h],$ i.e., we have $\C = \C(\{\gamma_{ij}^{(k)}\}).$

A failure pattern is given by a subset of  edges in the complete
bipartite graph $K_{m,n}$. For each $(i,j) \in E$ we have variables
$x_{ij}$, which are subject to parity check constraints at each
vertex. If a vertex has degree $1$, then the parity check lets us
solve for the corresponding variable. We iteratively eliminate such
vertices, until every remaining vertex in the graph has degree $2$ or
higher. Let $E$ denote the set of remaining failures and let $L, R$
denote the subset of vertices on the two sides that have non-zero
degree. Thus we have a bipartite graph $H(L,R,E)$ where $\deg(v) \geq
2$ for every vertex $v \in L \cup R$. Let $\ell = |L|, r = |R|, e =
|E|$ and let $c$ denote the number of connected components.

In the topology $T_{m \times n}(1,1,h)$, we will refer
to~\eqref{eq:col1} and~\eqref{eq:row1} as the XOR constraints, and
to~\eqref{eq:global1} as the global constraints.

\begin{Lem}
\label{lem:patterns}
Using the notation above a failure pattern $E$ is correctable by $T_{m \times n}(1,1,h)$ iff
\begin{equation}\label{Eqn:T111}
e \leq h + \ell + r -c.
\end{equation}
\end{Lem}
\begin{proof}
For every edge $e \in E$, we have a variable $x_e$. Let $e\sim v$
denote that $e$ is incident to $v$.  For every vertex $v \in L \cup R$
we have the constraint
\begin{align}
\label{eq:vertex}
\sum_{e \sim v}x_e =0
\end{align}
We will first show that the rank of the XOR constraints is $\ell + r -c$.

We start with the case when $c= 1,$ and $H(L,R,E)$ is connected. The upper bound on rank comes from observing that the $\ell + r$ linear constraints satisfy the dependence
\begin{align}
\sum_{v \in L}\sum_{e \sim v}x_e = \sum_{w \in R} \sum_{e \sim w}x_e
\end{align}
since every edge appears exactly once on the LHS and the RHS.
We claim that the constraints corresponding to any smaller subset $L' \cup R'$ of vertices are linearly independent. Indeed, we can rewrite a dependence between these constraints as
\begin{align}
\sum_{v \in L'}\sum_{e \sim v}x_e + \sum_{w \in R'} \sum_{e \sim w}x_e = 0.
\end{align}
But since $L' \cup R'$ does not induce a connected subgraph, there
must be at least one edge leaving the set, and the corresponding
variable appears exactly once. This proves that the rank of the parity
check constraints equals $(\ell + r -1)$ when $c= 1$.

When there are $c \geq 2$ connected components, each connected component
involves disjoint variables on the edges, hence constraints in
different components are trivially independent. So the bound of
$(\ell + r -c)$ follows by summing the above bound for each component.

We first consider the case when the number of unknowns $e$ satisfies
$e \leq (\ell + r -c) +h$. Here, any
subset of $e - (\ell +  r- c) \leq h$ global constraints in an MR
instantiation will be linearly independent of the local XOR
constraints. Hence these equations together with the local constraints
can recover the unknown variables, so the error pattern is
correctable.

When $e > (\ell + r - c) +h$, the total rank of all constraints is at most the RHS, hence they
are insufficient to recover all $e$ unknowns.
\end{proof}

Using this lemma, we extract a simple
sufficient condition for correctability. A simple cycle in $K_{m,n}$ is a connected subgraph where each vertex
has degree exactly $2$ (in other words, we do not repeat vertices in
the cycle).

\begin{Lem}
\label{Lem:sc}
  Let $E \subseteq [m] \times [n]$ be a failure pattern such that $H(L,R,E)$ is the union of $h$
vertex disjoint simple cycles. Then $E$ is correctable in $T_{m \times n}(1,1,h)$.
\end{Lem}
\begin{Proof}
Let $E =\cup_{i=1}^h C_i$ where the $C_i$s are vertex disjoint simple
cycles. We need to check the condition $e \leq h + \ell + r -c$. But $c =h$ since each cycle is a distinct component. There are a total of $\ell + r$ vertices, each of which has degree $2$, so $2e = 2(\ell + r)$.
Hence we in fact have $e = h + \ell + r -c$.
\end{Proof}

Simple cycles are in fact the only correctable patterns in $\tmn(1,1,1)$.

\begin{Lem}
Correctable failure patterns in $T_{m \times n}(1,1,1)$ correspond to
simple cycles in $K_{m,n}$.
\end{Lem}
\begin{Proof}
Since every vertex has degree at least $2$, we have
\[ e \geq \max(2\ell,2r) \geq \ell + r,\]
with strict inequality whenever some vertex has degree exceeding $2$. By Lemma \ref{lem:patterns},
\[ e \leq \ell + r -c +1 \leq \ell + r \]
with equality iff $c =1$. Thus if the error pattern $E$ is correctable, it is connected and every vertex in it has degree exactly $2$, so $E$ must be a simple cycle.
\end{Proof}

\subsection{Lower Bounds for $T_{n \times n}(a,b,h)$}
\label{Sec:mainLB}

We start by showing the lower bound for $T_{n \times n}(1,1,1)$
(Corollary \ref{Cor:LB} of Theorem \ref{Thm:LBgen}).

\begin{Cor}\label{Cor:Sim}
Let $\C = \C(\{\gamma_{e}\})_{e\in [m]\times [n]}$ be an instantiation of $T_{m
  \times n}(1,1,1)$.
The error pattern corresponding to a
simple cycle $C$ is correctable by $\C$ iff
\begin{align} \label{eq:correctable}
\sum_{e \in C} \gamma_e \neq 0.
\end{align}
\end{Cor}
\begin{proof}
Let $E$ be a simple cycle. The parity check constraints enforce the condition $x_e = x$ for every $e \in E$. Plugging this into the global parity gives
\[ \sum_{e\in E} \gamma_e x = x\sum_{e \in E}\gamma_e = 0.\]
Assuming that $\sum_{e \in E}\gamma_e  \neq 0$, the only solution to
this system is $x =0$. This shows that the system of equalities
defined by the variables $x_e$ and the parity check equations has a
trivial kernel, so it is invertible.
\end{proof}

Corollary~\ref{Cor:LB} now follows using Lemma~\ref{Lemma:Main}.
\begin{proofof}{of Corollary~\ref{Cor:LB}}
Consider an MR instantiation of $T_{n\times n}(1,1,1)$ over a field
$\mathbb{F}_q$ where $q=2^d.$ Use the global constraint to produce an assignment $\{\gamma_e\}$ of weights to edges of
$K_{n,n}.$ By Corollary~\ref{Cor:Sim} every simple cycle in $K_{n,n}$
now carries a non-zero weight. By Lemma~\ref{Lemma:Main} we have
$d\geq \Omega((\log n)^2).$ Thus $q\geq n^{\Omega(\log n)}.$
\end{proofof}

Next we consider $T_{m \times n}(1,1,h)$ for $h \leq n$ (for larger $h$, the bound claimed is trivial). We can consider MR instantions $\C = \C(\{\gamma_{ij}^{(k)}\})$ for $k \in [h]$ since $\alpha_i =1,
\beta_j =1$ for every $i \in [m], j \in [n]$. Let $\Gamma_{ij} \in
\F^h$ denote the vector $(\gamma_{ij}^{(k)})_{k \in [h]}$ of
coefficients associated with $x_{ij}$. For $S
\subseteq [n] \times [n]$, let $\Gamma(S) = \sum_{(i,j) \in E
}\Gamma_{ij}$ denote the sum of coefficient vectors over all indices
in the set $E$.

\begin{Cor}
  \label{Cor:Sim_h}
  In an MR instantiation of $\tmn(1,1,h)$, for any vertex-disjoint
  cycles $C_1,\ldots,C_h$, the vectors $(\Gamma(C_1),\ldots,\Gamma(C_h))$ are linearly independent.
\end{Cor}
\begin{Proof}
By Lemma \ref{Lem:sc},  the failure pattern $E = \cup_{t=1}^h C_t$ is correctable.
By the vertex disjointness, the parity check constraints imply that edge in the cycle $C_i$
carries the same variable $x_t$. Plugging this in the global parity equations gives
\[ \sum_{t=1}^h\Gamma(C_t)x_t = 0.\]
  The kernel is trivial iff  $(\Gamma(C_1),\ldots,\Gamma(C_h))$
  are linearly independent over $\F$.
\end{Proof}

\begin{Lem}
\label{lem:11h}
  Let $1 \leq h \leq n$. Any maximally recoverable code for the topology $T_{m \times
  n}(1,1,h)$ requires field size $q = \exp(\Omega((\log(\min(m,n)/h))^2))$.
\end{Lem}
\begin{Proof}
Assume that $m > n$, else we reverse their roles.
Let us partition $[n]$ into $h$ nearly equal parts
$P_1,\ldots,P_h$ of size at least  $\lfloor n/h \rfloor$ each.
We will consider sequences of $h$ simple cycles $(C_1,\ldots,C_h)$, where
$C_t$ only involves edges from $P_t \times P_t$. Note that these cycles
are vertex disjoint. We claim that for every $k \in [h]$, there exists $t \in [h]$ such that for every
  simple cycle $C_t$ with edges from $P_t \times P_t$, $\sum_{(i,j)
    \in C_t}\Gamma_{ij} \neq 0$.

  Assume for contradiction that $\exists k \in [h]$ so that $\forall
  t \in [h]$, $\exists C_t$ with edges from $P_t \times P_t$, such
  that $\sum_{(i,j) \in C_t}\gamma_{ij}^{(k)} = 0$. Now consider the error pattern $E =
  \cup_t C_t$. The vectors $\Gamma(C_1),\ldots,\Gamma(C_h) \in \F^h$ are all $0$
  in coordinate $k$, so they cannot be linearly independent. But this
  contradicts Corollary \ref{Cor:Sim_h}.

  Now consider $k =1$. There exists $P_t$ of size at least $\lfloor
  n/h \rfloor$ so that any simple cycle $C_t$ in $P_t \times P_t$
  satsifies $\sum_{(i,j) \in C_t}\gamma_{ij}^1 \neq 0$. The lower
  bound now follows from Lemma \ref{Lemma:Main}.
\end{Proof}

We now extend the proof to the case where $a, b \geq 1$, proving
Theorem \ref{Thm:LBgen} . We will assume that $ h \leq \min((m -a), (n-b))$, else the claim is trivial.

\begin{proofof}{of Theorem \ref{Thm:LBgen}}
Let $\C$ be an MR instantiation of
$T_{n \times n}(a,b,h)$. By Item (2) in Proposition
\ref{Prop:DualToPrimal} we can pick a subset $S$ of $[m -a] \times [n
  -b]$ of size $(m-a)(n-b) - h$, and let these be the information
symbols, while the remaining are parity checks symbols. We claim that puncturing
the code by restricting it to co-ordinates in $[m -a +1] \times [n- b
  +1]$ results in a code $\C'$ that is an MR instantiation of the topology $T_{(m-a +1)
  \times (n -b + 1)}(1,1,h)$. From this claim, the theorem follows by an
application of Lemma \ref{lem:11h}.

It is easy to see that $\C'$ does instantitate the topology $T_{(m-a +1)
  \times (n -b + 1)}(1,1,h)$. We will prove that it is MR by
contradiction. Assume that some failure pattern $E' \subseteq [m -a
  +1] \times [n- b +1]$ is not correctable for $\C'$ but is corrected
by some instantiation $\C''$ of $T_{(m-a +1)  \times (n -b +
  1)}(1,1,h)$. We extend $E'$ to a failure pattern $E$ in $[m]
\times [n]$ by adding all the puntured co-ordinates to $E'$. This
resulting failure pattern is corrected by any instantiation of $T_{m
  \times n}(a,b,h)$ whose puncturing to $[m -a] \times [n -b]$ is
$\C''$, but not by $\C$ since decoding $E$ using $\C$ reduces to
correcting $E'$ using $\C'$. This contradicts the assumption that $\C$ is MR.
\end{proofof}

\section{Characterizing correctable patterns in $T_{m \times n}(1,b,0)$}
\label{Sec:Class}

In this section, we will prove Theorem~\ref{Th:Class} in three steps:
\begin{enumerate}
\item Lemma~\ref{Lemma:OmitLow} shows that it suffices to consider erasure patterns $E$ where every non-empty row of $E$ has at least $b+1$ erasures.
\item Lemma~\ref{Lemma:BplusOne} establishes the Theorem for erasure patterns $E$ where every non-empty row has weight exactly $b+1$.
\item Lemma~\ref{Lemma:BoostingReg} extends the proof to general regular erasure patters.
\end{enumerate}

Note that by Lemma~\ref{Lemma:RegNecess} we only need to argue that regularity is sufficient for MR correctability. To do this, for every regular erasure pattern $E\subseteq [m]\times [n]$ we need to exhibit a column code $C_{\mathrm{col}}$ and a row code $C_{\mathrm{row}}$, so that $C_{\mathrm{col}}\otimes C_{\mathrm{row}}$ corrects $E$. Indeed, we can tailor the choice of these codes to the pattern $E$, and our proof will use the flexibility.

Let $E\subseteq [m]\times [n]$ be an erasure pattern for $T_{m\times n}(1,b,0).$ For $i\in [m],$ we refer to $\left(\{i\}\times [n]\right)\cap E$ as the $i$-th row of $E.$ We often call the number of elements in the $i$-th row the {\it weight} of the row.

\begin{lemma}\label{Lemma:OmitLow}
Let $E$ be an erasure pattern for $T=T_{m\times n}(1,b,0).$ Suppose $E^\prime \subseteq E$ is obtained from $E$ by restricting $E$ to rows where $E$ has $b+1$ or more erasures; then $E$ is correctable for $T$ iff $E^\prime$ is correctable for $T.$
\end{lemma}
\begin{proof}
Clearly, if $E$ is correctable for $T$ then $E^\prime \subseteq E$ is also correctable for $T$. We need to show the converse. So assume that code $C$ instantiating the topology $T$ corrects $E^\prime$. We can assume that $C$ is maximally recoverable since a maximally recoverable code for this topology will also correct $E^\prime$. By Proposition \ref{Prop:DualToPrimal}, each row of $C$ is an MDS code capable of correcting $b$ erasures. So we can use row decoding to correct all rows that have $b$ or fewer erasures, which reduces the problem to correcting $E^\prime$. By assumption, $C$ can correct $E^\prime$, and hence $E$.
\end{proof}

Below is the main technical lemma of this Section.

\begin{lemma}\label{Lemma:BplusOne}
Let $E$ be a regular pattern for $T=T_{m\times n}(1,b,0).$ Suppose that every row of $[m]\times [n]$ that intersects $E,$ intersects $E$ in exactly $b+1$ locations; then $E$ is correctable for $T.$
\end{lemma}
\begin{proof}
We fix $C_{\mathrm{col}}$ to be the simple parity code, i.e., we set all $\{\alpha_i^{(1)}\}, i\in [m]$ in~(\ref{Eqn:Top1}) to one, and focus on constructing the code $C_{\mathrm{row}}.$ Let $U\times V,$ $|U|=u,$ $ |V|=v$ be the smallest enclosing sub-grid for $E.$ In what follows we often find it convenient to represent $E$ by the bipartite graph $G$ with node set $U\cup V$ and edge set representing $E$ in the natural way. By~(\ref{Eqn:RegIntersect}) we have
\begin{equation}
|E| = u(b+1) \leq v+(u-1)b.
\end{equation}
Thus $u\leq v-b.$ Let $d=(v-b)-u.$ We set $C_{\mathrm{row}}$ to be the linear space spanned by $(n-v)+(u+d) = n-b$ vectors: $(n-v)$ unit vectors $\{{\bf e}_i\}_{i\in [n]\setminus V}$ and $u+d$ vectors ${\bf w}_1,\dots,{\bf w}_u,{\bf z}_1,\ldots,{\bf z}_d \in \mathbb{F}^n,$ over some large finite field $\mathbb{F}.$ Note that this constitutes a valid choice of the row code as the co-dimension is necessarily at least $b.$ We constrain vectors $\{{\bf w}_i\}$ and $\{{\bf z}_i\}$ to have no support outside of $V.$ Therefore we often treat these vectors as elements of $\mathbb{F}^v$ rather than $\mathbb{F}^n.$ Furthermore, for every $i\in [u],$ we constrain ${\bf w}_i$ to have no support outside of the support of the $i$-th row of $E.$  Let $M\in \mathbb{F}^{(u+d)\times v}$ be the matrix whose rows are vectors $\{{\bf w}_i\}$ and $\{{\bf z}_i\}.$
We pick the field $\mathbb{F}$ to be sufficiently large and select $\{{\bf w}_i\}$ and $\{{\bf z}_i\}$ at random from $\mathbb{F}^v$ subject to the support constraints. This allows us to ensure that every minor of $M$ that can have full rank for some choice of $\{{\bf w}_i\}$ and $\{{\bf z}_i\}$ indeed has full rank. In particular for all $U^\prime\subseteq U$ and $V^\prime \subseteq V$ such that there is a matching of size $|U^\prime|$ in $G$ between the nodes $U^\prime$ and $V^\prime$ the minor $M_{U^\prime, V^\prime}$ has full rank. Also, all coordinates of $\{{\bf w}_i\}_{i\in [u]}$ and $\{{\bf z}_i\}_{i\in [d]}$ that do not have to be zero are non-zero.

Below is the key Claim that underlies our proof:

{\it Fix $j\in [u]$ and consider an arbitrary linear combination ${\bf y}$ of vectors $\{{\bf w}_i\}_{i\in [u]\setminus \{j\}}$ and $\{{\bf z}_i\}_{i\in [d]}$ that includes at least one of these vectors with a non-zero coefficient. We claim that
\begin{equation}\label{Eqn:SuppNonCon}
\mathrm{supp}({\bf y})\not\subseteq \mathrm{supp}({\bf w}_j).
\end{equation}
}
We first prove the claim above and then proceed with the proof of the Lemma. Assume~(\ref{Eqn:SuppNonCon}) is violated. Let $U^\prime = U\setminus \{j\}$ and $V^\prime = V\setminus \mathrm{supp}({\bf w}_j).$ It is possible to take a non-trivial linear combination of $\{{\bf w}_i\}_{i\in U^\prime}$ and $\{{\bf z}_i\}_{i\in [d]}$ that has no support in $V^\prime.$ Observe that $|U^\prime|+d=u-1+d=v-(b+1)=|V^\prime|.$ Therefore existence of a linear combination as above implies that a certain $(u-1+d)\times |V^\prime|$ square minor $M^\prime$ of $M$ is degenerate.

By discussion preceding~(\ref{Eqn:SuppNonCon}), we conclude that the restriction of the graph $G$ to nodes $(U^\prime,V^\prime)$ has no matching of size $U^\prime,$ as any such matching together with the fact that vectors $\{z_i\}$ are random with full support could be used to imply that $M^\prime$ is of full rank. Thus by Hall's theorem~\cite[p.55]{Jukna} there exists a set $U^{\prime\prime}\subseteq U^\prime$ such that the size of the neighborhood $N(U^{\prime\prime})$ in $G$ is at most $|U^{\prime\prime}|-\Delta$ for a positive $\Delta.$ Let
$U^{\prime\prime\prime}=U^{\prime\prime}\cup \{j\}$ and $V^{\prime\prime\prime}=N(U^{\prime\prime})\cup \mathrm{supp}({\bf w}_j).$ Further, let $t=|U^{\prime\prime\prime}|.$ We have \begin{equation}\label{Eqn:ViolReg1}
|E\cap(U^{\prime\prime\prime}\times V^{\prime\prime\prime})| = t\cdot (b+1).
\end{equation}
However
\begin{equation}\label{Eqn:ViolReg2}
|V^{\prime\prime\prime}| + (|U^{\prime\prime\prime}|-1)\cdot b= (b+1+(t-1)-\Delta)+(t-1)b=t(b+1)-\Delta.
\end{equation}
Thus restricting $E$ to $U^{\prime\prime\prime}\times V^{\prime\prime\prime}$ violates~(\ref{Eqn:RegIntersect}). This contradiction completes the proof of the Claim. We now use the Claim to prove the Lemma.

Assume for the purpose of contradiction that $C_{\mathrm{col}}\otimes C_{\mathrm{row}}$ does not correct $E.$ Then $C_{\mathrm{col}}\otimes C_{\mathrm{row}}$ contains a codeword ${\bf w}$ such that $\mathrm{supp}({\bf w})\subseteq E.$ We now make two observations:
\begin{itemize}
\item {\it For all $i\in U,$ the restriction of ${\bf w}$ to row $i$ has be a scalar multiple of the vector ${\bf w}_i$ defined above.} This observation follows from the fact that the $i$-th row of ${\bf w}$ is an element of $C_{\mathrm{row}}$ and by the Claim above no element of $C_{\mathrm{row}}$ other than the scalar multiplies of ${\bf w}_i$ has its support inside $\mathrm{supp}({\bf w}_i).$
\item {\it Vectors $\{{\bf w}_i\}_{i\in U}$ are linearly independent.} Again this easily follows for the Claim. Every dependency between $\{{\bf w}_i\}_{i\in U}$ can be used to obtain a linear combination of $\{{\bf w}_i\}_{i\in U\setminus \{j\}}$ whose support falls within $\mathrm{supp}({\bf w}_j)$ for some $j\in U.$
\end{itemize}
By the first bullet above, rows of ${\bf w}$ are scalar multiples of vectors $\{{\bf w}_i\}_{i \in U}.$ However, rows of ${\bf w}$ are linearly dependent as every column of $C_{\mathrm{col}}\otimes C_{\mathrm{row}}$ is an element of $C_{\mathrm{col}}.$ We obtain a contradiction with the second bullet above. This completes the proof of the Lemma.
\end{proof}

Lemma~\ref{Lemma:BplusOne} suffices to establish Theorem~\ref{Th:Class} for erasure patterns whose weight is $b+1$ across all non-empty rows. We now reduce the general case to this special case. In what follows we often use the same character to denote a row of a topology and the set of erased coordinates of that row. Our reduction is based on the following definition.

\begin{definition}\label{Def:Boosting}
Let $E\subseteq [m]\times [n]$ be an erasure pattern for $T_{m\times n}(1,b,0).$ Assume that non-empty rows of $E$ have weights $b+r_1,\ldots,b+r_u,$ where all $\{r_i\}_{i\in [u]}$  are positive. Set $\delta = \sum_{i\in [u]} (r_i - 1).$ We define the {\it boosting} of $E$ to be an erasure pattern $\mathcal B (E)$ for $T_{(m+\delta) \times n}(1,b,0)$ where $\mathcal B (E)$ is obtained via the following process:
\begin{itemize}
\item Each row of $T_{m\times n}(1,b,0)$ that does not intersect $E$ yields a row in $T_{(m+\delta) \times n}(1,b,0)$ that does not intersect $\mathcal B(E)$.
\item Each row ${\bf s}_i$ of $T_{m\times n}(1,b,0)$ that intersects $E$ in $b+r$ coordinates is replaced by $r$ rows ${\bf s}_{i1},\ldots, {\bf s}_{ir},$ where every $\{{\bf s}_{ij}\}_{j\in [r]}$ contains the first $b$ elements of ${\bf s}_i;$ the weight of each $\{{\bf s}_{ij}\}_{j\in [r]}$ is $b+1;$ and the union of supports of all $\{{\bf s}_{ij}\}_{j\in [r]}$ is the support of ${\bf s}_i.$
\end{itemize}
\end{definition}

We demonstrate the concept of boosting by the following example.
\begin{example}\label{Example:Boosting}
A pattern $E$ for $T_{2\times 5}(1,2,0)$ and the boosted pattern ${\mathcal B}(E)$ for $T_{4\times 5}(1,2,0).$
\begin{equation}
E=
\left(
\begin{array}{ccccc}
1 & 1 & 1 &   & 1 \\
1 &   & 1 & 1 & 1 \\
\end{array}
\right)
\ \
\Rightarrow
\ \
{\mathcal B}(E) =
\left(
\begin{array}{ccccc}
1 & 1 & 1 &   &   \\
1 & 1 &   &   & 1 \\
1 &   & 1 & 1 &   \\
1 &   & 1 &   & 1 \\
\end{array}
\right)
\end{equation}
\end{example}

The following lemma shows that boosting preserves regularity.
\begin{lemma}\label{Lemma:BoostingReg}
Let $E\subseteq [m]\times [n]$ be an erasure pattern for $T=T_{m\times n}(1,b,0)$ where every non-empty row has weight $b+1$ or more. Let $E^\prime={\mathcal B}(E)$ be the boosting of $E$ viewed as an erasure pattern for $T^\prime=T_{m^\prime\times n}(1,b,0).$ If $E$ is regular; then $E^\prime$ is also regular.
\end{lemma}
\begin{proof}
Let $U^\prime \times V$ be an arbitrary sub-grid of $T^\prime.$ Let $|U^\prime|=u^\prime,$ $|V|=v.$ Note that rows of $T^\prime$ that arise by boosting the rows of $E$ have two indices $(i,j).$ Let $U=\{i\in [m] \mid \exists j : (i,j) \in U^\prime\}.$ Let $|U|=u.$ Note that $U\times V$ is a sub-grid of $T.$ In what follows we argue that $E^\prime$ does not violate~(\ref{Eqn:RegIntersect}) on $U^\prime\times V;$ since $E$ does not violate~(\ref{Eqn:RegIntersect}) on $U\times V.$ Consider
\begin{equation}\label{Eqn:Reg1}
\begin{array}{lll}
\Delta(E^\prime) & = & |E^\prime\cap \left(U^\prime \times V \right)| - (v+(u^\prime -1) b)   \\
                 & = & \sum\limits_{(i,j)\in U^\prime} \left(|{\bf s}_{ij}\cap V|-b\right)-(v-b) \\
                 & = & \sum\limits_i\sum\limits_{j:(i,j)\in U^\prime}\left(|{\bf s}_{ij}\cap V|-b\right)-(v-b).
\end{array}
\end{equation}
We claim that for all $i\in [m]:$
\begin{equation}\label{Eqn:Reg2}
\sum\limits_{j:(i,j)\in U^\prime}\left(|{\bf s}_{ij}\cap V|-b\right)\leq |{\bf s}_{i}\cap V|-b.
\end{equation}
To see this assume that $V$ intersects the sets of $b$ first elements of ${\bf w}_i$ in $c\leq b$ coordinates. Let the sum above include $t$ terms. Expression on the left simplifies to $tc+t^\prime -tb,$ for $t^\prime\leq t.$ Expression on the right simplifies to $c+t^\prime-b.$ It remains to note that
$$
tc+t^\prime -tb\leq c+t^\prime-b\quad \mathrm{as}\quad c\leq b.
$$
and~(\ref{Eqn:Reg2}) follows. Combining~(\ref{Eqn:Reg1}) and~(\ref{Eqn:Reg2}) we conclude that
$$
\Delta(E^\prime) \geq \sum\limits_i \left(|{\bf s}_{i}\cap V|-b\right) -(v-b) = \Delta(E)\geq 0.
$$
This completes the proof of the Lemma.
\end{proof}

\begin{proofof}{of Theorem~\ref{Th:Class}}
By Lemma~\ref{Lemma:OmitLow} we can assume that every row of $E$ has weight at least $b+1.$ Consider the boosted pattern $\mathcal B (E)$ for $T^\prime=T_{m^\prime\times n}(1,b,0).$ By Lemma~\ref{Lemma:BoostingReg}, $\mathcal B(E)$ is regular. Thus by Lemma~\ref{Lemma:BplusOne}, $\mathcal B(E)$ is correctable for $T^\prime.$ Let $C_{\mathrm{col}}\otimes C_{\mathrm{row}}$ be the instantiation that corrects $\mathcal B(E)$ obtained via Lemma~\ref{Lemma:BplusOne}. Note that $C_{\mathrm{col}}$ is the simple parity check code that we denote by $P_{m^\prime}.$

We claim that the tensor product of the parity check code $P_m$ with $C_{\mathrm{row}}$ corrects $E$ for $T.$ Assume the contrary. Let ${\bf w}$ be the codeword of $P_m\otimes C_{\mathrm{row}}$ where $\mathrm{supp}({\bf w})\subseteq E.$ Let $u$ be the number of non-zero rows in $E.$ For $i\in [u],$ let ${\bf s}_i$ be the $i$-th non-zero row of ${\bf w}.$ Assume for all $i\in [u],$  the weight of ${\bf s}_i$ is $b+r_i.$ We now use ${\bf w}$ to obtain a codeword ${\bf w}^\prime$ that resides on $\mathcal B(E)$ for $P_{m^\prime}\otimes C_{\mathrm{row}}$ instantiation of $T^\prime.$ Our construction is based on the following observation:
\begin{equation}\label{eqn:DimRestr}
\dim \left(C_{\mathrm{row}}|_{\mathrm{supp}({\bf w}_i)}\right) = r_i.
\end{equation}
We now prove the observation. Firstly, consider vectors $\{{\bf w}_{ij}\}_{j\in [r_i]}$ with supports $\{{\bf s}_{ij}\}_{j\in [r_i]}$ that are used in the construction of the linear space $C_{\mathrm{row}}$ in the proof of Lemma~\ref{Lemma:BplusOne}. These $r_i$ vectors are clearly linearly independent given their support structure. Secondly, note that if $C_{\mathrm{row}}$ had any vector, other then the linear combinations of $\{{\bf w}_{ij}\}_{j\in [r_i]},$ reside on $\mathrm{supp}({\bf w}_i);$ then we would immediately get a contradiction to the key Claim~(\ref{Eqn:SuppNonCon}) in the proof of Lemma~\ref{Lemma:BplusOne}.

Using the observation and the argument above we can represent every ${\bf w}_i$ as a unique linear combination of vectors $\{{\bf w}_{ij}\}_{j\in [r_i]}$ with supports $\{{\bf s}_{ij}\}_{j\in [r_i]}.$ Now the collection of vectors $\{{\bf w}_{ij}\}_{i\in [u], j\in [r_i]}$ yields a codeword ${\bf w}^\prime$ of $P_{m^\prime}\otimes C_{\mathrm{row}}$ that resides entirely on $\mathcal B(E),$ contradicting the fact that $\mathcal B(E)$ is correctable for~$T^\prime.$
\end{proofof}

\section{Maximally recoverable codes over linear fields for $T_{m \times n}(1,0,2)$}\label{Sec:H2}

Recall that correctable patterns for this topology are fully characterized~\cite{BHH,GHJY}:

\begin{lemma}\label{Lemma:SimSpec}
A pattern $E\subseteq [m]\times [n]$ is correctable for $T_{m\times
  n}(1,0,h)$ iff it can be obtained by erasing at most one coordinate in every column of $[m]\times [n]$ and then additionally up to $h$ more coordinates.
\end{lemma}

We now prove Theorem~\ref{Th:H2}.

\begin{proofof}{of Theorem~\ref{Th:H2}}
Let $M$ be the smallest power of $2$ that is no less than $m$ and $N$ be the smallest power of $2$ that is no less than $n.$ We now present an explicit MR instantiation of $T_{m\times n}(1,0,2)$ over $F_{MN}:$
\begin{itemize}
\item We set all $\left\{\alpha_i^{(1)}\right\}_{i\in [m]}$ in~(\ref{Eqn:Top1}) to be equal to one.
\item To complete the specification of the code we need to specify $\left\{\gamma_{ij}^{(1)}, \gamma_{ij}^{(2)}\right\}_{i\in [m], j\in [n]}\in F_{MN}.$ In order to do so let us fix a set $\{s_1,s_2,\ldots,s_m\}\subseteq\mathbb{F}_{MN}$ to be a subset of an additive subgroup $G\subseteq\mathbb{F}_{MN}$ of size $M$ and $c_1,\ldots,c_n\in\mathbb{F}_{MN}$ to be field elements, such that $c_{j_1}\not \in c_{j_2}+G,$ for $j_1\ne j_2.$ (In other words $\{c_j\}_{j\in [n]}$ belong to different cosets of $\mathbb{F}_{MN}$ modulo the subgroup $G.$) For $i\in [m], j\in [n]$ we set
    \begin{equation}\label{Eqn:102}
    \begin{array}{lll}
    \gamma_{ij}^{(1)} & = & s_i,          \\
    \gamma_{ij}^{(2)} & = & s_i^2 + c_j\cdot s_i.  \\
    \end{array}
    \end{equation}
\end{itemize}

By Lemma~\ref{Lemma:SimSpec} it suffices to show that every pattern of erasures obtained by erasing one location per column and two more arbitrary locations is correctable by our instantiation of $T_{m\times n}(1,0,2).$ Note that every column that carries just one erasure easily corrects this erasure since all $\alpha_i^{(1)} = 1.$ We consider two cases:
\begin{itemize}
\item {\it Some column $j\in [n]$ carries three erasures.} Assume erasures are in rows $i_1,i_2$ and $i_3.$ Solving linear system~(\ref{Eqn:Top1}),~(\ref{Eqn:Top3}) amounts to inverting a $3\times 3$ matrix, whose determinant is non-zero:
$$\det\begin{pmatrix}
1 & 1 & 1 \\
s_{i_1} & s_{i_2} & s_{i_3} \\
s_{i_1}^2+c_j\cdot s_{i_1} & s_{j_2}^2+c_j\cdot s_{j_2} & s_{i_3}^2+c_j\cdot s_{i_3}
\end{pmatrix}=\det\begin{pmatrix}
1 & 1 & 1 \\
s_{i_1} & s_{i_2} & s_{i_3} \\
s_{i_1}^2 & s_{i_2}^2 & s_{i_3}^2 \\
\end{pmatrix}\neq 0.$$
Therefore the erasure pattern is correctable.

\item {\it There are two distinct columns $j_1,j_2 \in [n]$ each carrying two erasures.} Assume column $j_1$ has erasures in rows $i_1$ and $i_2,$ while column $j_2$ has erasures in rows $i_3$ and $i_4.$ This time solving linear system~(\ref{Eqn:Top1}),~(\ref{Eqn:Top3}) amounts to inverting a $4\times 4$ matrix, whose determinant is again non-zero:
\begin{align*}
&\det\begin{pmatrix}
1 & 1 & 0 & 0 \\
0 & 0 & 1 & 1 \\
s_{i_1} & s_{i_2} & s_{i_3} & s_{i_4} \\
s_{i_1}^2+c_{j_1}\cdot s_{i_1} & s_{i_2}^2+c_{j_1}\cdot s_{i_2} & s_{i_3}^2+c_{j_2}\cdot s_{i_3} & s_{i_4}^2+c_{j_2}\cdot s_{i_4}
\end{pmatrix}\\
=&\det\begin{pmatrix}
1 & 0 & 0 & 0 \\
0 & 0 & 1 & 0 \\
s_{i_1} & s_{i_1}+s_{i_2} & s_{i_3} & s_{i_3}+s_{i_4} \\
s_{i_1}^2+c_{j_1} s_{i_1} & (s_{i_1}+s_{i_2})^2+c_{j_1}(s_{i_1}+s_{i_2}) & s_{i_3}^2+c_{j_2} s_{i_3} & (s_{i_3}+s_{i_4})^2+c_{j_2}(s_{i_3}+s_{i_4})
\end{pmatrix} \\
=&\det\begin{pmatrix}
s_{i_1}+s_{i_2} & s_{i_3}+s_{i_4} \\
(s_{i_1}+s_{i_2})^2+c_{j_1}(s_{i_1}+s_{i_2}) & (s_{i_3}+s_{i_4})^2+c_{j_2}(s_{i_3}+s_{i_4})
\end{pmatrix} \\
=&(s_{i_1}+s_{i_2})(s_{i_3}+s_{i_4})(s_{i_1}+s_{i_2}+s_{i_3}+s_{i_4}+c_{j_1}+c_{j_2})\neq0.
\end{align*}
In the last step, we used the fact that $s_{i_1}+s_{i_2}+s_{i_3}+s_{i_4}+c_{j_1}+c_{j_2}\neq0$. This follows from $s_{i_1}+s_{i_2}+s_{i_3}+s_{i_4}\in G$ and $c_{j_2}\not \in c_{j_1}+G.$
\end{itemize}
\end{proofof}

\section{Open problems}\label{Sec:Open}
The theory of maximally recoverable codes is in its infancy. There is a wide array of questions that remain open. Here we highlight some of the prominent ones:
\begin{enumerate}
\item The topology $T_{m\times n}(1,0,h)$ is well studied in the
  literature and used in practice. Yet the alphabet size of MR codes
  for this topology is poorly understood. There is a linear
  $\Omega(mn)$ lower bound that applies when $h\geq 2.$ For $h=2$ this
  bound is asymptotically tight by Theorem~\ref{Th:H2}. For $h\geq 3$
  the best constructions~\cite{GHJY} use alphabet of size
  $O((mn)^{ch)})$ for constants $c < 1$. Obtaining a super-linear
  lower bound or improving the upper bound would be of great
  interest.

\item Establish Conjecture~\ref{Conj:ClassMain} regarding correctable
  error patterns in $T_{m \times n}(a,b,0)$ for  $a>1$. What are the
  correctable erasure patterns for $T_{m\times n}(a,b,h)$ for general
  $h > 0$?

\item We do not know how to construct Maximally Recoverable tensor
  product codes, even for the special case of $T_{m\times n}(1,b,0)$
  where we now know a classification of correctable failure patterns.

\item For $T_{n \times n}(1,1,h)$,   there is also an explicit  construction \cite[Theorem 31]{GHJY} that gives
MR codes for this topology over fields of size $\exp(O(n \log n))$,
matching the probabilistic

 construction. This is still  from our
lower bound of $\exp((\log(n))^2)$. Closing this gap is an interesting
open problem. Can one obtain a lower bound of
  $\exp\left(n^{\Omega(1)}\right)$ for some $T_{n \times n}(a,b,h)?$

\item Can one generalize Corollary~\ref{Cor:LB} to get a lower bound
  of $\exp(\log m\cdot \log n)$ for $T_{m\times n}(1,1,1)?$ Our
  current bound is $\exp((\min(\log m, \log n))^2)$.

\end{enumerate}

\bibliography{code-locality-Aug}

\newpage

\appendix

\section{Related Work}\label{Sec:Work}

The first family of codes with locality for applications in storage comes from~\cite{HCL}. That paper also introduced the concept of maximal recoverability, in a restricted setting which does not allow for locality among parity check symbols. In this restricted setting, they gave a combinatorial characterization of correctable failure patterns via Hall's theorem.

The approach of using a topology to ensure local recovery from correlated failures has been studied in the literature~\cite{BHH,Blaum,GHJY,BPSY} and is used in practice~\cite{HuangSX,MLRH14}. The first definition of maximal recoverability for an arbitrary topology was given~\cite{GHJY}.

The work of~\cite{GHSY} introduced a formal definition of locality,
and focused on codes that guarantee locality for a single failure. For
this simple setting, they were able to show that optimal codes must
have a certain natural topology. Maximally recoverable codes for that
topology had been studied earlier in the work of~\cite{BHH,Blaum}
where they were called PMDS codes. The best known general
constructions are due to~\cite{GHJY}. A construction of codes over
fields of linear size with the largest possible distance for
$T_{m\times n}(1,0,h)$ (a weaker property than maximal recoverability)
has been given in~\cite{TB}.  Discussion regarding the importance of
using small finite fields in codes for storage can be found
in~\cite[Section 2]{Plank}.

The study of codes with locality and in particular maximally recoverable codes is distantly related to the study of Locally Decodable Codes (LDCs)~\cite{Y_now}. The key differences are as follows: LDCs can be viewed as codes where every symbol has low locality even after a constant-fraction of codeword coordinates are erased. The main challenge is to minimize the codeword length of these codes given the locality constraints. Instead, MR codes only provide locality after certain structured failures, the layout of which is known at the stage of code design and that are few in number. Codeword length is fixed by specifying the topology, and the key challenge is to minimize the field size while providing optimal erasure correction guarantees.

The porblem addressed in Lemma \ref{Lemma:Main} can be viewed as an instance of the {\em critical problem} of Crapo and Rota from the 70s~\cite{CR}, where the goal is to find the largest dimensional subspace in $\F_2^N$ that does not intersect  a given set $S \subset \F_2^N$, which generalizes the problem of finding the maximum rate binary linear code. Identify $[N]$ with the edges of  $K_{n,n}$. Given $\gamma:[N] \to \F_2^d$, the indicators of all sets of edges $E$ such that $\sum_{e \in E}\gamma_e =0$ is a subspace, of dimension $N -d$ or more. Our goal is to find the largest such subspace that does not intersect the set $S \subset \F_2^N$ of indicators of all simple cycles.

Another related problem had been recently studied in~\cite{FGT} in the
context of derandomizing parallel algorithms for matching. The authors
also consider the problem of assigning weights to edges of a graph, so
that simple cycles carry non-zero weight. The key differences from our
setting are: we need a single assignment while~\cite{FGT} may have
multiple assignments; we care about all simple cycles,
while~\cite{FGT} only needed non-zero weights on short cycles; we are
interested in fields of characteristic $2$ while~\cite{FGT} work in
characteristic zero.

Our work on MR codes for $T_{m\times n}(a,b,0)$  bears some
similarities to the study of weight heirarchies of product
codes~\cite{SW,WYang}. The difference is that there one is interested
in understanding the relation between weight hierarchies of codes and
their tensor products, while we are concerned with minimizing the
field size of codes whose tensor products have optimal erasure
correction capabilities.

\section{More Proofs}
\label{sec:app}

\begin{proof}[Proof of Lemma \ref{Lemma:RegNecess}]
Let $U\times V$ be a sub-grid of $[m]\times [n],$
where~(\ref{Eqn:RegIntersect}) is violated. Let $|U|=u$ and
$|V|=v$. Consider the collection of variables $\{x_{ij}\}$ from
Definition~\ref{Def:Topology}. Let us restrict our attention to
$\{x_{ij}\}$ where $i\in U$ and $j\in V$ and set all other $x_{ij}$ to
zero. By the last bullet in Section~\ref{Sec:Notation}, the rank of
the row and column constraints on the variables $\{x_{ij}\}_{i\in U,
  j\in V}$ is at most $va + ub -ab$.

We next set all $x_{ij}$s outside of the set $E$ to zero. Setting variables outside $E$ to $0$ can only reduce the rank of the row and column contraints further. The  number of surviving variables is  $|E| > ua +vb -ab$, since ~(\ref{Eqn:RegIntersect}) is violated, which exceeds the rank of the constraints. So there exists a codeword supported on $E$, and $E$ is not correctable by any code $C$ that instantiates $T_{m\times n}(a,b,0)$.
\end{proof}

\begin{proof}[Proof of Proposition \ref{Prop:DualToPrimal}]
We proceed item by item.
\begin{enumerate}
\item We first argue that $\dim C \geq (m-a)(n-b)-h.$ To see this note that constraints from groups one and two in Definition~\ref{Def:Topology} yield a tensor product of linear codes $C_{\mathrm{col}}$ and $C_{\mathrm{row}},$ where $\dim C_{\mathrm{col}}\geq m-a$ and $\dim C_{\mathrm{row}} \geq n-b.$ We have,
    \begin{equation}\label{Eqn:ProdDim}
    \dim C_{\mathrm{col}}\otimes C_{\mathrm{row}} = \dim C_{\mathrm{col}}\cdot \dim C_{\mathrm{row}}\geq (m-a)(n-b).
     \end{equation}
     Adding group 3 constrains can reduce the dimension by no more than $h.$

    Now assume $\dim C > (m-a)(n-b)-h.$ Consider another instantiation $C^\prime$ of $T_{m \times n}(a,b,h)$ where $\dim C^\prime_{\mathrm{col}}=m-a,$ $\dim C^\prime_{\mathrm{row}}=n-b,$ and thus constraints from groups one and two yield a code of dimension $(m-a)(n-b).$ Let $S\subseteq [m]\times [n]$ be the information set for that code. Set constrains in group 3 to be linearly independent and have no support outside of $S.$ This implies $\dim C^\prime = (m-a)(n-b)-h.$ Let $E\subseteq [m]\times [n]$ be a complement of the information set of $C^\prime.$ It is easy to see that $C^\prime$ recovers $E.$ However $C$ cannot recover $E$ as $\dim C > mn-|E|.$

    Therefore $\dim C = (m-a)(n-b)-h \geq \dim C_{\mathrm{col}}\cdot \dim C_{\mathrm{row}} - h.$ Thus $\dim C_{\mathrm{col}}=m-a$ and $\dim C_{\mathrm{row}}=n-b.$

\item Fix $C^\prime_{\mathrm{col}}$ and $C^\prime_{\mathrm{row}}$ to be MDS codes. Clearly, $U$ is an information set of $C^\prime_{\mathrm{col}}$ and $V$ is an information set of $C^\prime_{\mathrm{row}}.$ Therefore $U\times V$  is an information set of $C^\prime_{\mathrm{col}}\otimes C^\prime_{\mathrm{row}}.$ To complete the specification of $C^\prime,$ fix constrains in group 3 to have no support outside of $U\times V$ and define an MDS code of co-dimension $h$ on $U\times V.$ Now $C^\prime|_{U\times V}$ is an MDS code of dimension $(m-a)(n-b)-h.$ Thus $C|_{U\times V}$ also has to be an MDS code of dimension $(m-a)(n-b)-h,$ as $C$ corrects every pattern of erasures that is corrected by $C^\prime.$

\item Let $C$ be an arbitrary MR instantiation. Our goal here is to show that under the mild technical assumption~(\ref{Eqn:TechLoc}) both column restrictions of $C$ and row restrictions of $C$ have to be MDS codes of dimensions (respectively) $m-a$ and $n-b$\footnote{One can show that in general the converse is not true; a tensor product of two MDS codes is not necessarily maximally recoverable.}.

    If suffices to prove this claim for column codes. Let $C_{\mathrm{col}}^\prime$ be the restriction of $C$ to column $j.$ Clearly, we have $C_{\mathrm{col}}^\prime \subseteq C_{\mathrm{col}}.$ Thus $\dim C_{\mathrm{col}}^\prime\leq C_{\mathrm{col}} = m-a,$ where the latter identity follows from item 1 above. Observe that if we show that $C_{\mathrm{col}}^\prime$ is an $[m,m-a,a+1]$ code, this would in particular imply that $C_{\mathrm{col}}^\prime=C_{\mathrm{col}}.$

    Note that $C_{\mathrm{col}}^\prime$ is an $[m,m-a,a+1]$ code if and only if every $(m-a)$-sized subset its coordinates is an information set. Assume that there exists some subset $U\subseteq [m],$ $|U|=m-a$ that is not an information set of $C_{\mathrm{col}}^\prime.$ In particular, there exists a linear dependence between the symbols of $C_{\mathrm{col}}^\prime|_U.$ Let $V\subseteq [n],$ $|V|=n-b$ be arbitrary. By the item above, for all $h$-sized sets $H\subseteq U\times V,$ the set $\left(U\times V\right)\setminus H$ is an information set for $C.$ By~(\ref{Eqn:TechLoc}) it is possible to pick $H$ so that $\left(U\times V\right)\setminus H$ contains a complete column of $U\times V.$ In such case $\left(U\times V\right)\setminus H$ cannot be an information set as entries of the column are linearly dependent, and we arrive at a contradiction.
\end{enumerate}
\end{proof}

\end{document}